\documentclass[preprint,12pt]{elsarticle}

\usepackage[utf8]{inputenc}
\usepackage[T1]{fontenc}
\usepackage{graphicx}
\usepackage{amssymb}
\usepackage{amsmath}
\usepackage{amsfonts}
\usepackage{amsthm}
\usepackage{floatrow}
\usepackage{subcaption}
\usepackage{titlesec}
\journal{Journal of Parallel and Distributed Computing}

\theoremstyle{definition}
\newtheorem{defn}{Definition}[section]
\theoremstyle{plain}
\newtheorem{thm}{Theorem}[section]
\newtheorem{lem}[thm]{Lemma}

\begin{document}

\begin{frontmatter}

\title{Protocol for Asynchronous, Reliable, Secure and Efficient Consensus (PARSEC) Version 2.0}

\author[1]{Pierre Chevalier}
\author[1]{Bartłomiej Kamiński}
\author[1]{Fraser Hutchison}
\author[1]{Qi Ma}
\author[1]{Spandan Sharma}
\author[1,2]{Andreas Fackler}
\author[3]{William J Buchanan}

\address[1]{MaidSafe Ltd.}
\address[2]{POA Networks Ltd.}
\address[3]{Edinburgh Napier University}

\begin{abstract}
In this paper we present an open source, fully asynchronous, leaderless algorithm for reaching consensus in the presence of Byzantine faults in an asynchronous network. We prove the algorithm's correctness provided that less than a third of participating nodes are faulty. We also present a way of applying the algorithm to a network with dynamic membership, i.e. a network in which nodes can join and leave at will. The core  contribution of this paper is an optimal model in the definition of an asynchronous BFT protocol, and which is resilient to 1/3 byzantine nodes. This model matches an agreement with probability one (unlike some probabilistic methods), and where a common coin is used as a source of randomization so that it respects the FLP impossibility result.

\end{abstract}

\end{frontmatter}

\section{Introduction}

This paper presents a new Byzantine fault tolerant (BFT) consensus algorithm that can work under asynchronous conditions. Like Hashgraph \cite{baird2016swirlds} and Aleph \cite{gkagol2018aleph}, it has no leaders, no round robin, no proof-of-work and reaches eventual consensus with probability one. It is also fully open, and a work-in-progress implementation written in Rust is available \cite{parsec}. Like HoneyBadger BFT \cite{miller2016honey}, this algorithm is built by composing a number of good ideas present in the literature. A gossip protocol is used to allow efficient communication between nodes \cite{guerraoui2010lifting}, as in Hashgraph \cite{baird2016swirlds}, Aleph \cite{gkagol2018aleph} and Avalanche \cite{rocket2018snowflake}.

The general problem of reaching Byzantine agreement on any value is reduced to the simpler problem of reaching binary Byzantine agreement on the nodes participating in each decision. This allows us to reuse the elegant binary Byzantine agreement protocol described in \cite{aba} after adapting it to the gossip protocol.

Finally, even though a trusted dealer (or another trusted external source of private key shares) is still required to initialise the instances of the algorithm (like in \cite{miller2016honey}), changes to the set of processes executing the instances (which we call \emph{dynamic membership}) can be made without the need for such external sources. The resulting algorithm is a Protocol for Asynchronous, Reliable, Secure and Efficient Consensus. PARSEC is a key building block of the SAFE Network, an ethical decentralized network of data and applications providing Secure Access For Everyone \cite{penland2016towards}.

The key contribution of this paper is the creation of an asynchronous BFT protocol, and which is resilient to 1/3 byzantine nodes. This is an optimal model. It also satisfies an agreement with probability one (unlike some probabilistic methods), and with a common coin for source of randomization that it respects the FLP impossibility result \cite{borowsky1993generalized}\cite{fischer1982impossibility}. The algorithm uses gossip for efficient and resilient communication. It is leaderless (with a caveat at network start-up) and the paper outlines how it can be adapted it to a dynamic membership context, while remaining leaderless after start-up. In comparable work, HBBFT (Honey Badger of BFT) has a less efficient communication mechanism with a secure broadcast \cite{miller2016honey}, and HashGraph is not rigorous in a liveness proof \cite{alpern1985defining} and where the published work only discusses the possibility of using a common coin in passing \cite{baird2016swirlds}. AlephZero \cite{gkagol2018aleph} is more recent than our initial work \cite{parsec} and is an improvement on HashGraph, and which also includes the use of a common coin. Avalanche \cite{rocket2018snowflake} differs from our work in that is uses a synchronous context. None of these methods, though, includes dynamic membership, and this is a key differentiator in this paper.

\section{Related work}

A gossip protocol has been likened to office workers spreading a rumour, and where Alice starts a new rumour, and then passes it to Bob, who then passes it to Dave. Alice then tells it to Frank, and who might have already heard it from Dave. In this way the rumour propagates quickly through a network, depending on the frequency that those spreading the rumours will pass them on.  The advantages of gossip protocols was outlined by \cite{birman2007promise} and who defined that they could be used with autonomic self-management, repair of inconsistencies, reliable multicast and distributed search. In actual operation, a gossip protocol involves a group of agents who hold private information, and who can communicate with each other. The core objective is for all of the agents to learn the private information \cite{apt2017computational}, and where distributed epistemic methods lead to simplified system models. A core weakness of the gossip protocols is that the dissemination of the rumour might not be radiated across a whole network of connections. Recent applications of gossip protocols has included the verifying the consistency of certificate logs \cite{chuat2015efficient}.

Honey Badger is Byzantine Fault Tolerant \cite{miller2016honey} which is asynchronous in its scope. It does not involve a leader node beyond the trusted setup phase and can cope with corrupted nodes. It does not actually make any commitments around the timing of the delivery of a message, and where even if Eve control the scheduling of messages, there will be no impact of the overall consensus. It can reach a consensus within an infrastructure of $f$ failed nodes, and where the total number of node ($N$) is greater than $3\times f$. 

Hashgraph builds on the Directed Acyclic Graphs (DAG) approach \cite{baird2016swirlds} with gossip communications; \emph{gossip-about-gossip} methods; and Byzantine voting algorithms. The gossip-about-gossip method allows for the history of all the communications within an infrastructure to be reconstructed. Avalanche \cite{rocket2018snowflake} focuses on a scaleable electronic payment system. It uses synchronous communications with a leaderless BFT. In the face of adversaries, it uses a probabilistic safety guarantee.

\section{The algorithm description}

\subsection{The network model}

We assume the network to be a set $\mathcal{N}$ of $N$ instances of the algorithm communicating via asynchronous connections, which means that messages sent over these connections are all delivered eventually, but we make no assumptions regarding the delays between sending the message and its reception. In such a setup, it is impossible to distinguish between an instance failing by completely stopping and a large delay in message delivery.

We allow a possibility of up to $f$ Byzantine (arbitrary) failures, where $3f < N$. We will call
the instances that have not failed \emph{correct} or \emph{honest}, and the failing instances
\emph{faulty} or \emph{malicious} - as the Byzantine failure model allows for malicious behaviour
and collaboration. We will refer to any set of instances containing more than $\frac{2}{3}N$ of them as a \emph{supermajority}.

\subsection{Data structures}

A node executing the algorithm keeps two data structures: a \emph{gossip graph} and an
ordered set of \emph{blocks}. The vertices of the gossip graph, called \emph{gossip events},
contain the following fields:

\begin{itemize}
		\item Payload - data the node wants to pass to other nodes
		\item Self-parent (optional) - a cryptographic hash of another gossip event created by the
			same node
		\item Other-parent (optional) - a hash of another gossip event created by some other node
		\item Cause - cause of creation for this event; can be \emph{sync}, \emph{observation} or
			\emph{coin share}
		\item Creator ID - the public key of the event's creator
		\item Signature - a cryptographic signature of the above fields
\end{itemize}

The self-parent and other-parent are always present, except for the first events created by
respective nodes, as there are no parent events to be referred to in such cases. Other-parent is
also absent in events created because of an observation or a coin share - because there is no
gossip partner in such a case.

The blocks in the ordered set are network events signed by a subset of the nodes in the network.
This set is the output of the algorithm, and represents an order of network events that all nodes
agree upon. We call the blocks that are elements of the ordered set \emph{stable blocks}. Let us also define a few useful terms regarding the gossip graph for future use.

\begin{defn}
	We say that event $A$ is an \emph{ancestor} of event $B$ iff: $A = B$, or $A$ is an ancestor of
	$B$'s self-parent, or $A$ is an ancestor of $B$'s other-parent.
\end{defn}

\begin{defn}
	We say that event $A$ is a \emph{self-ancestor} of event $B$ iff: $A = B$, or $A$ is a
	self-ancestor of $B$'s self-parent.
\end{defn}

\begin{defn}
	We say that event $A$ is a \emph{descendant} of event $B$ iff $B$ is an ancestor of $A$.
\end{defn}

\begin{defn}
	We say that event $A$ is a \emph{self-descendant} of event $B$ iff $B$ is a self-ancestor of
	$A$.
\end{defn}

Following Hashgraph\cite{baird2016swirlds}, we also define two additional useful notions:

\begin{defn}
	An event $A$ is said to \emph{see} an event $B$ iff $B$ is an ancestor of $A$, and there
	doesn't exist any pair of events by $B$'s creator $B_1$, $B_2$, such that $B_1$ and $B_2$ are
	ancestors of $A$, but $B_1$ is neither an ancestor nor a descendant of $B_2$ (see fig.
	\ref{fig-seen}). We call a situation in which such a pair exists a \emph{fork}.
\end{defn}

\begin{defn}
	An event $A$ is said to \emph{strongly see} an event $B$ iff $A$ sees a set of events created
	by a supermajority of nodes in the system that all see $B$ (see fig. \ref{fig-stronglyseen}).
\end{defn}

\begin{figure}[!ht]
	\centering
	\begin{floatrow}
		\ffigbox[\FBwidth]{%
			\caption{d\_4 sees b\_0: b\_0 is its ancestor and there are no forks}
			\label{fig-seen}}{%
			\includegraphics[width=.4\textwidth]{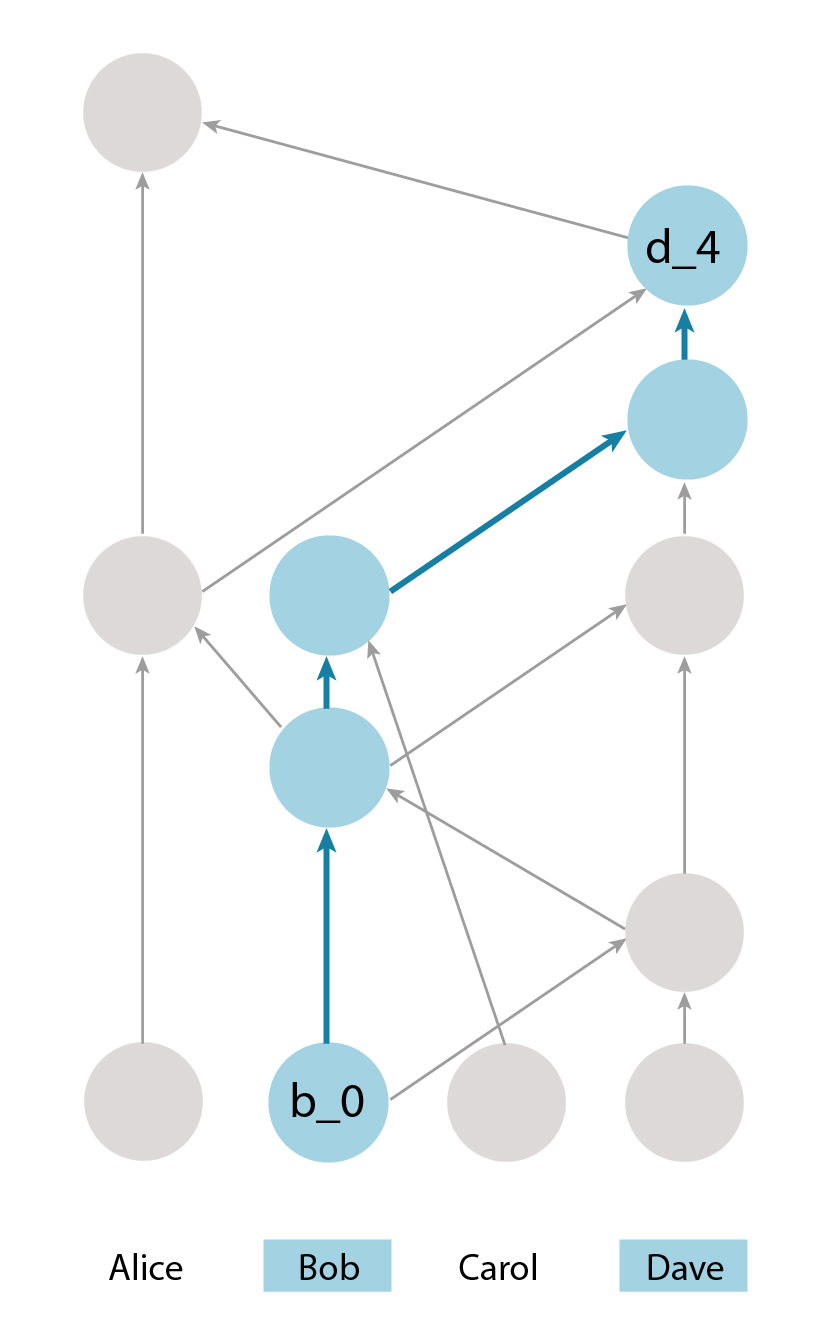}
		}
		\ffigbox[\FBwidth]{%
			\caption{a\_1 strongly sees b\_0: it sees itself, b\_1 and d\_1, which have been
			created by a supermajority of nodes and all see b\_0}
			\label{fig-stronglyseen}}{%
			\includegraphics[width=.4\textwidth]{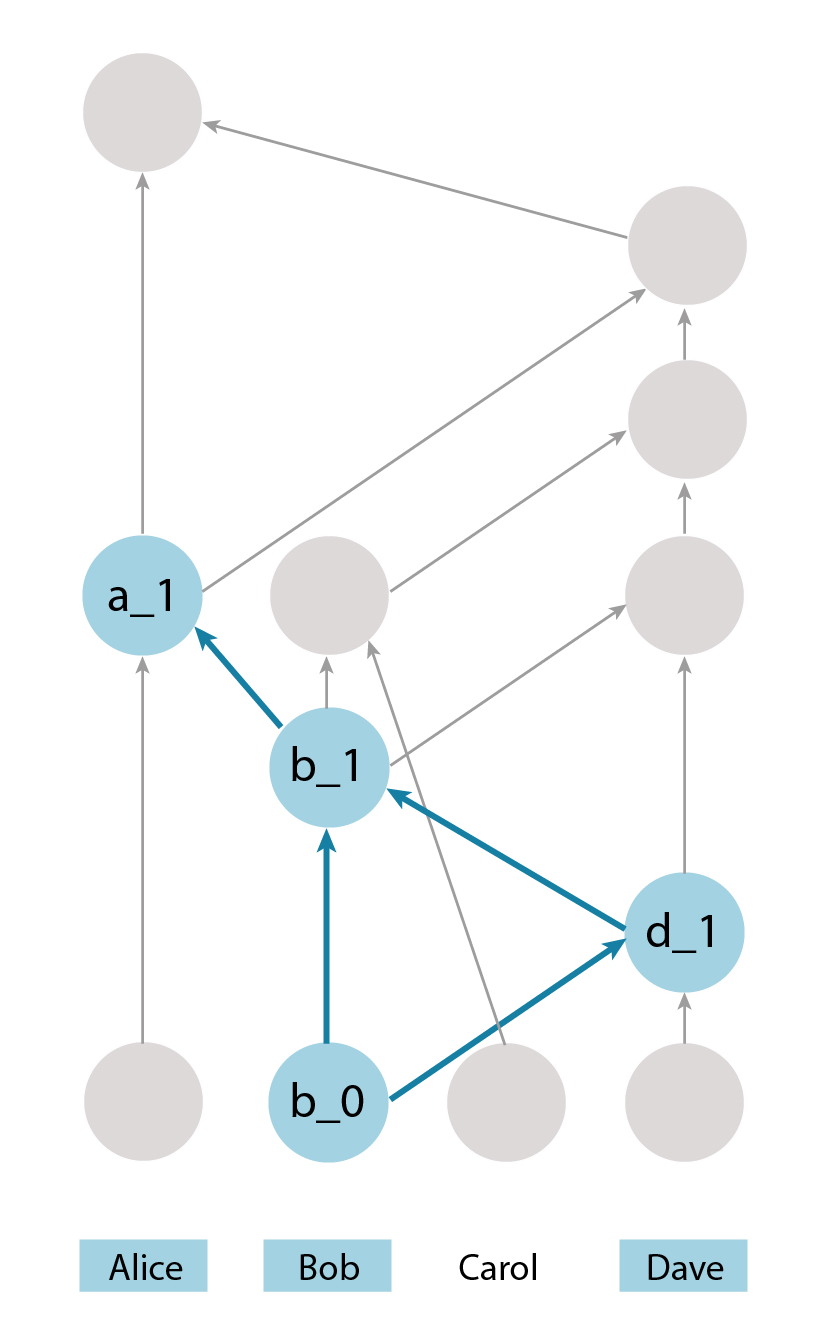}
		}
	\end{floatrow}
\end{figure}

\subsection{General overview of the algorithm}

The nodes execute two main steps in an infinite loop:

\begin{itemize}
		\item Synchronise the gossip graph with another random node
		\item Determine whether any new blocks should be appended to the ordered set
\end{itemize}

\subsubsection{Synchronisation}

This step is responsible for building the gossip graph and spreading information around the
network. Nodes continually make random calls, called \emph{sync calls}, to other nodes and exchange
information about the graph, so that all correct nodes end up with the same data in their graphs.
The hashes and signatures in gossip events make sure that malicious nodes won't be able to tamper
with any part of the graph.

Whenever a node receives a sync call, it creates a new gossip event. The self-parent of this event
is the hash of the last gossip event created by the recipient, and the other-parent is the hash of
the last event created by the sender (which the recipient learns about from the call). The new event
also stores the reason for which it was created (cause: \emph{sync}).

If the recipient of a sync call believes it knows a network event that should be appended as the
next one in the chain, it records its vote as the payload of the newly created event. The other
nodes will learn of this vote during subsequent sync calls made by its creator.

\begin{lem}\label{descendants}
	If A and B are correct nodes, then every event created by A will eventually have a descendant
	created by B.
\end{lem}

\begin{proof}
	This trivially follows from the network assumption that every message is eventually delivered, and
	from the fact that nodes continue to make sync calls, which result in the callee creating a
	descendant of the caller's last event.
\end{proof}

\subsubsection{Determining order}

During this step, a node analyses the graph, counts the votes and decides which block should become the next one. This step is a complex one and so it is described in detail in a separate subsection below.

\subsection{Calculating the order}

To be able to order blocks, we need first to have some blocks that can be ordered.

Every event has a set of \emph{interesting payloads} associated with it, which are some of the
network events its ancestors contain votes for. The exact way the set of interesting payloads is
calculated is left to the user; however, it has to satisfy some constraints:

\begin{itemize}
	\item If event $e$ has an interesting payload $p$, there exists an ancestor of $e$ containing a vote for $p$.
	\item If $e$ has an interesting payload $p$, no self-descendants of $e$ have an interesting
		payload $p$.
	\item If event $a$ has an interesting payload $p$, $a$ is an ancestor of event $e$ and $e$ has
		no self-ancestor having an interesting payload $p$, then $e$ has an interesting payload $p$.
	\item If $e$ has an interesting payload $p$, then no stable block contains $p$.
\end{itemize}

For example, $p$ might become an interesting payload of $e$ if $e$ has a single ancestor that
contains a vote for $p$. Another option is that $p$ only becomes an interesting payload if there is a supermajority of ancestors of $e$ containing votes for $p$. An event that has a non-empty set of interesting payloads is called an \emph{interesting event}.

From the first constraint on interesting payloads it follows that an interesting event always has a self-parent. Only the initial events have no self-parents, but they are their only ancestors, and they never contain votes, so they can't be interesting.

The first gossip event created by any given node which strongly sees interesting events created by a supermajority of nodes is said to be an \emph{observer}. The interesting events don't need to have the same interesting payloads - in fact, it is the case when they have different payloads that is the most interesting.

Since observers are descendants of interesting events, and interesting events always have
self-parents, it follows that observers always have self-parents as well. An observer implicitly carries a list of $N$ \emph{meta-votes}. Every meta-vote is just a binary value denoting whether a corresponding node's interesting event is to be taken into account when determining the order. An observer casts a meta-vote of $true$ on a node if it can strongly see an interesting event by that node. Each node casts a meta-vote on every node, hence each node casts $N$
meta-votes, and since an observer strongly sees a supermajority of interesting events, by
definition, more than $\frac{2}{3}N$ of them are $true$.

Meta-votes reduce the problem of Byzantine agreement about the order to that of binary Byzantine
agreement, which has been solved previously\cite{aba}. The algorithm described in \cite{aba}, like many ABFT algorithms, requires a device called a common coin. PARSEC is no different in this regard. In PARSEC, we utilise a common coin based on a threshold cryptography scheme using Boneh-Lynn-Shacham signatures \cite{bls}.

\subsubsection{Binary agreement}

For the sake of simplicity, we will define the algorithm in terms of deciding a single
\emph{meta-election} - that is, deciding whether or not to take a single node's opinion into
account when trying to choose a single new block. We can view a meta-election for node $X$ with
latest agreed block $B$ as a function on a subset $H_{X,B}$ of the gossip graph $G$, which is the
set of all events that are descendants of any observer of this meta-election:

\[ \mathsf{meta\_election}_{X,B}: H_{X,B} \to \{0, 1, \bot\} \]

The $\bot$ value means that the result has not been decided yet at this point in the graph.

Any gossip event which is an element of $H_{X, B}$ and is not an observer trivially has a
self-parent in $H_{X, B}$.

From this point on, until section \ref{nextblock}, whenever we mention a meta-election, we mean a
single meta-election regarding a single node, with a specific block $B$ being the last stable one.

In order to calculate the meta-election value for events in $H_{X,B}$, we will need to calculate a
few helper values as well:

\begin{itemize}
	\item $\mathsf{stage}$ - a counter denoting the calculation stage
	\item $\mathsf{estimates}$ - a set of one or two values estimating the final result
	\item $\mathsf{bin\_values}$ - a helper set of binary values
	\item $\mathsf{aux}$ - a helper binary value
\end{itemize}

$\mathsf{stage}$ is an integer value which represents the stage of the protocol we are considering
when looking at a specific gossip event. A number is associated with each gossip event, such that
the $\mathsf{stage}$ of the observers is always 0. The $\mathsf{stage}$ of any other gossip event
is either the $\mathsf{stage}$ of its self-parent, or the stage of its self-parent plus one under
specific conditions. The exact conditions under which the stage is incremented will be described
later in more details. Other variables such as $\mathsf{estimates}$, $\mathsf{bin\_values}$ and
$\mathsf{aux}$ all depend on the stage.

$\mathsf{estimates}$ (abbreviated $\mathsf{est}$) is a set of binary values that represent the
perceived opinion(s) of the creator of any gossip event on the outcome of a meta-election. The
initial $\mathsf{estimates}$ of an observer is the set containing just its own meta-vote, and it is
the set of the self-parent's $\mathsf{estimates}$ for other events, except the events which increase
the stage - it is then calculated from the results of the previous stage.

If the initial estimates for an event contain a single value $v$, and that event sees more than
$\frac{N}{3}$ events with $\neg v$ in their $\mathsf{estimates}$ (which means that at least one
honest node estimated $\neg v$), this opposite value gets added to its $\mathsf{estimates}$ (so it
will contain both true and false).

Note: the convention in function definitions below is that the value of the function is the first
value for which the corresponding condition is satisfied. We also use the common convention of $0$
denoting \emph{false}, and $1$ denoting \emph{true}.

\[ \mathsf{init\_est}: H_{X,B} \to 2^{\{0,1\}} \]
\[ \mathsf{init\_est}(e) = \left\{ \begin{array}{ll}
	\{ v \} & \textrm{if $e$ is an observer} \\
	& \textrm{with meta-vote $v$} \\
	\mathsf{next\_est}(\mathsf{self\_par}(e)) & \textrm{if $\mathsf{stage}(e) >
		\mathsf{stage}(\mathsf{self\_par}(e))$} \\
	\mathsf{est}(\mathsf{self\_par}(e)) & \textrm{otherwise}
\end{array} \right. \]

\[ \mathsf{est}: H_{X,B} \to 2^{\{0,1\}} \]
\[ \mathsf{est}(e) = \left\{ \begin{array}{ll}
	\{ v \} & \textrm{if there exists an ancestor $d$ of $e$} \\
	& \textrm{such that $v = \mathsf{meta\_election}(d) \neq \bot$} \\
	\{ 0,1 \} & \textrm{if $\mathsf{init\_est}(e) = \{v\}$} \\
	& \textrm{and $e$ sees $\geq\frac{N}{3}$ events $x$}\\
	& \textrm{by different nodes such that} \\
	& \textrm{$\mathsf{stage}(x) = \mathsf{stage}(e)$ and $\neg v \in \mathsf{est}(x)$} \\
	\mathsf{init\_est}(e) & \textrm{otherwise}
\end{array} \right. \]

$\mathsf{self\_par}(e)$ denotes $e$'s self-parent, and $\mathsf{next\_est}$ and $\mathsf{stage}$
will be defined later, once we have defined more values related to the events.

Once an event can see a supermajority of events by different nodes which agree in their estimates,
this agreed estimate becomes an element of this event's $\mathsf{bin\_values}$ (abbreviated
$\mathsf{bv}$). This set serves to validate values proposed by other nodes - if they propose
something we don't have in $\mathsf{bin\_values}$, we will reject it, as we have no way to ensure
its validity.

\[ \mathsf{bv}: H_{X,B} \to 2^{\{0,1\}} \]
\[ \mathsf{bv}(e)  =  \{ v: \begin{array}[t]{l} \textrm{there exist $> \frac{2}{3}N$ events $x$} \\
	\textrm{by different nodes such that} \\
	\textrm{$e$ sees $x$ and $\mathsf{stage}(e) = \mathsf{stage}(x)$ and $v \in \mathsf{est}(x)$} \\
	\textrm{or there exists an ancestor $d$ of $e$} \\
	\textrm{such that $\mathsf{meta\_election}(d) = v \neq \bot$}\}
\end{array}\]

\begin{figure}[!ht]
	\centering
	\includegraphics[width=\textwidth]{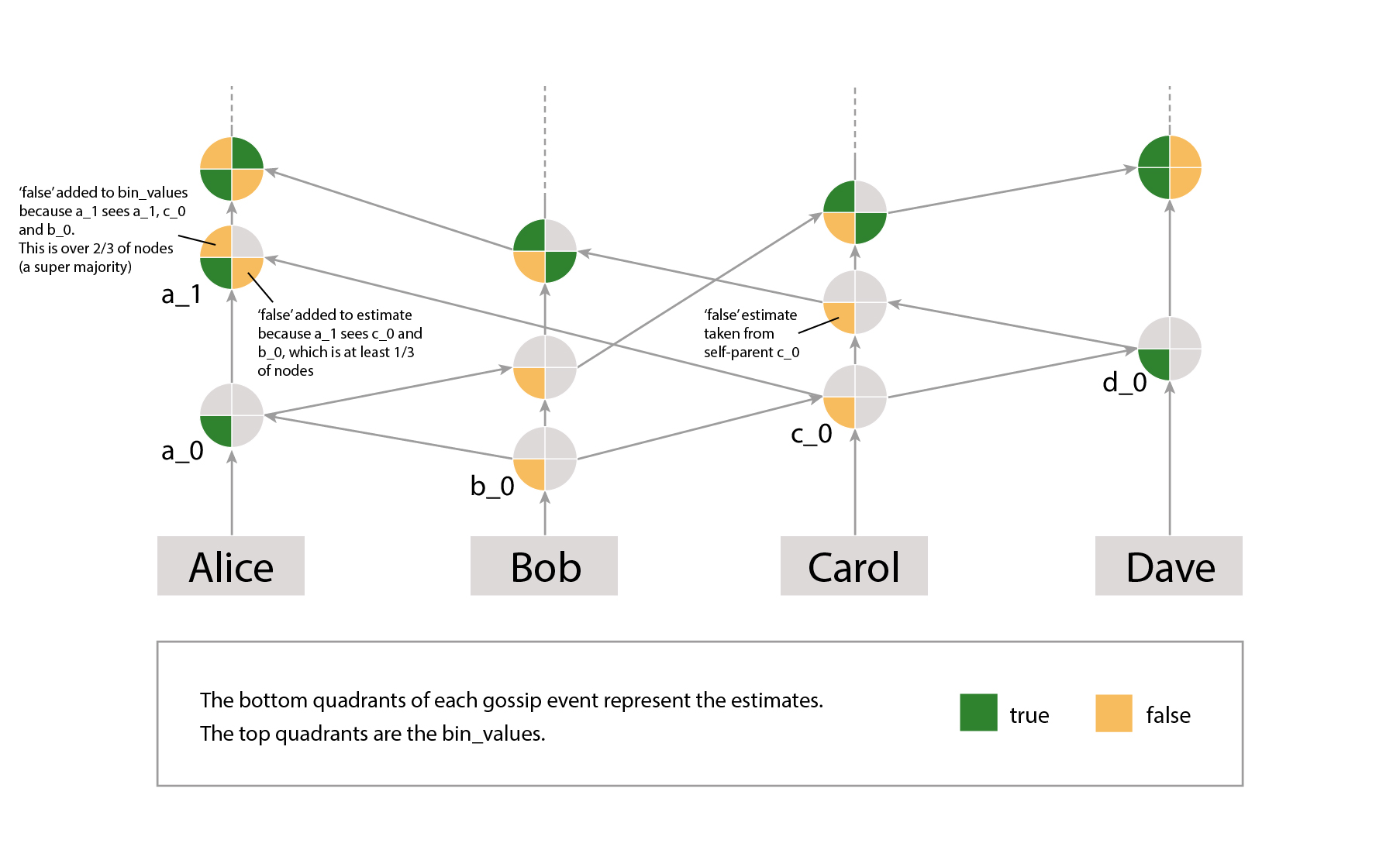}
	\caption{An example gossip graph, along with estimates and bin\_values
	associated with each gossip event. It illustrates how different nodes process the information
	they receive in order to populate their bin\_values.}
\end{figure}

If an event's parent has an empty $\mathsf{aux}$ value, and the event itself has non-empty
$\mathsf{bin\_values}$, it can propose a value to be agreed. This proposing is realised by having a
non-empty $\mathsf{aux}$ value. If $\mathsf{bin\_values}$ contains just one value, this value
becomes the $\mathsf{aux}$ value; otherwise, we can pick an arbitrary value, so we will pick true.
If the parent's value isn't empty, it becomes our value as well.

\[ \mathsf{aux}: H_{X,B} \to \{0, 1, \bot\} \]
\[ \mathsf{aux}(e) = \left\{ \begin{array}{ll}
	v & \textrm{if there exists an ancestor $d$ of $e$} \\
	& \textrm{such that $v = \mathsf{meta\_election}(d) \neq \bot$} \\
	\bot & \textrm{if $\mathsf{bv}(e) = \varnothing$} \\
	w & \textrm{if $\mathsf{bv}(e) = \{ w \}$} \\
	& \textrm{and $\mathsf{aux}(\mathsf{self\_par}(e)) = \bot$ } \\
	1 & \textrm{if $\mathsf{bv}(e) = \{ 0,1 \}$} \\
	& \textrm{and $\mathsf{aux}(\mathsf{self\_par}(e)) = \bot$ } \\
	\mathsf{aux}(\mathsf{self\_par}(e)) & \textrm{if $\mathsf{aux}(\mathsf{self\_par}(e))
		\neq \bot$}
\end{array}\right. \]

Whenever an event sees a supermajority of events with valid $\mathsf{aux}$ values, we perform the
common coin protocol (described in section \ref{commoncoin}), which will lead either to deciding
the final agreed value, or updating the estimates and moving to the next stage.

First, let us define some helper functions:

\[ \mathsf{supermajority\_valid\_aux}: H_{X,B} \to \{0,1\} \]
\[ \mathsf{supermajority\_valid\_aux}(e) = \begin{array}[t]{l}
	\textrm{$e$ sees a supermajority} \\
	\textrm{of events $x$ by different nodes} \\
	\textrm{such that $\mathsf{stage}(x) = \mathsf{stage}(e)$} \\
	\textrm{and $\mathsf{aux}(x) \in \mathsf{bv}(e)$}
\end{array} \]

\[ \mathsf{count\_aux}: H_{X,B} \times \{0, 1\} \to \mathbb{N} \]
\[ \mathsf{count\_aux}(e, v) = \begin{array}[t]{l}
	\textrm{number of events $x$ by different nodes such that} \\
	\textrm{$e$ sees $x$ and $\mathsf{stage}(x) = \mathsf{stage}(e)$} \\
	\textrm{and $\mathsf{aux}(x) \in \mathsf{bv}(e)$ and $\mathsf{aux}(x) = v$}
\end{array} \]

Now we can define how to determine a decided value:

\[ \mathsf{meta\_election}: H_{X,B} \to \{0, 1, \bot\} \]
\[ \mathsf{meta\_election}(e) = \left\{ \begin{array}{ll}
	v & \textrm{if there exists an ancestor $d$ of $e$} \\
	& \textrm{such that $v = \mathsf{meta\_election}(d) \neq \bot$} \\
	1 & \textrm{if $\mathsf{coin\_flip}(e) = 1$} \\
	& \textrm{and $\mathsf{count\_aux}(e, 1) > \frac{2}{3}N$} \\
	0 & \textrm{if $\mathsf{coin\_flip}(e) = 0$} \\
	& \textrm{and $\mathsf{count\_aux}(e, 0) > \frac{2}{3}N$} \\
	\bot & \textrm{otherwise}
\end{array} \right. \]

$\mathsf{coin\_flip}$ is the value of the common coin flip and will be defined later, in section
\ref{commoncoin}.

If an event sees a supermajority of valid $\mathsf{aux}$ values, but isn't able to decide, the next
event will mark the beginning of the next stage of the algorithm. This lets us finally define
$\mathsf{stage}$:

\[ \mathsf{stage}: H_{X,B} \to \mathbb{N} \]
\[ \mathsf{stage}(e) = \left\{ \begin{array}{ll}
	0 & \textrm{if $e$ is an observer} \\
	\mathsf{next\_stage}(\mathsf{self\_par}(e)) & \textrm{otherwise}
\end{array}\right. \]
\[ \mathsf{next\_stage}: H_{X,B} \to \mathbb{N} \]
\[ \mathsf{next\_stage}(e) = \mathsf{stage}(e) + \left\{ \begin{array}{ll}
	1 & \textrm{if $\mathsf{supermajority\_valid\_aux}(e)$} \\
	& \textrm{and $\mathsf{next\_est}(e) \neq \bot$} \\
	0 & \textrm{otherwise}
\end{array}\right.\]

If we don't decide in a stage, we need new estimates for the next one. This is being taken care of
by the common coin protocol briefly mentioned before, and described in more detail in section
\ref{commoncoin}.

In every stage, when the coin value is known, we can either decide or calculate the initial estimate
for the next stage. The general rule is this: we decide $v$ if we see a supermajority of
$\mathsf{aux}$ values of $v$ and the coin value is $v$. If we see a supermajority of $\mathsf{aux}$
values of $\neg v$ and the coin value is $v$, we estimate $\neg v$ in the next stage. If we don't
see any supermajority, we estimate the coin value in the next stage.

To calculate new estimates, we will define a $\mathsf{next\_est}$ function (which appeared already
in the definition of $\mathsf{est}$):

\[ \mathsf{next\_est}: H_{X,B} \to 2^{\{0,1\}} \cup \{ \bot \} \]
\[ \mathsf{next\_est}(e) = \left\{ \begin{array}{ll}
	\{v\} & \textrm{if $\mathsf{count\_aux}(e, v) > \frac{2}{3}N$} \\
	& \textrm{and $\mathsf{coin\_flip}(e) \neq \bot$} \\
	\{\mathsf{coin\_flip}(e)\} & \textrm{if $\mathsf{count\_aux}(e, 0) \leq \frac{2}{3}N$} \\
	& \textrm{and $\mathsf{count\_aux}(e, 1) \leq \frac{2}{3}N$} \\
	& \textrm{and $\mathsf{coin\_flip}(e) \neq \bot$} \\
	\bot & \textrm{otherwise}
\end{array} \right. \]

We can now start defining the $\mathsf{coin\_flip}$ function.

\subsubsection{Common coin}\label{commoncoin}

The common coin protocol is used to calculate the value of the coin flip in a stage. Every stage can
either have a predefined coin value, or demand a genuine flip. The exact pattern of which stage has
which option associated with it can be defined by the user, as long as there will be infinitely
many genuine flips with the stage number tending towards infinity (otherwise the termination
property of the algorithm - explained later in the paper - will not hold).

The simplest pattern would be to require a genuine flip at every stage. But other patterns
\cite{miller}, e.g. 1, 0, flip, 1, 0, flip\cite{trivial}, etc. can be used to avoid some of the
expensive flips, and optimise more for the optimistic case. Using fixed values at some stages may
speed up reaching consensus in some cases, as it doesn't require the exchange of coin shares and
returns the coin value right away.

For the genuine flip, the nodes do need to exchange coin shares. In order to define what they are,
let us first define the \emph{round hash} as follows. This hash will help us uniquely identify one
particular stage of a specific meta-election for a given node's meta vote:

\[ \mathsf{round\_hash}: H_{X,B} \to [0, 2^{256}) \]
\[ \mathsf{round\_hash}(e) = \mathsf{hash}( \mathsf{hash}(X), \mathsf{hash}(payload(B)),
	\mathsf{hash}(\mathsf{stage}(e))) \]

All nodes are assumed to possess private key shares - parts of a Boneh-Lynn-Shacham private key in a
threshold scheme, in which at least $\frac{N}{3}$ signature shares (signatures generated with
private key shares) are needed to reconstruct a full signature, and any $\frac{N}{3}$ signature
shares will result in the same, bit-by-bit identical signature. All nodes are also in possession of
public keys corresponding to all other nodes, so that they can verify each signature share
independently, as well as the full signature.

When a node creates an event $e$ with $\mathsf{stage}(e)$ corresponding to a genuine flip, that sees
a supermajority of aux values in its stage, it signs $\mathsf{round\_hash}(e)$ with its private key
share and publishes the resulting signature share in a gossip event with cause \emph{coin share}.
Once an event sees enough events with valid coin shares, it can collect them and calculate the full
signature, of which the lowest order bit will be taken as the coin flip value.

Let us define $\mathsf{count\_shares}$ analogously to $\mathsf{count\_aux}$:

\[ \mathsf{count\_shares}: H_{X,B} \times \{0, 1\} \to \mathbb{N} \]
\[ \mathsf{count\_shares}(e) = \begin{array}[t]{l}
	\textrm{number of events $x$ by different nodes such that} \\
	\textrm{$e$ sees $x$ and $\mathsf{stage}(x) = \mathsf{stage}(e)$} \\
	\textrm{and $x$ contains a valid coin share for $\mathsf{stage}(e)$} \\
\end{array} \]

The genuine flip can then be defined as:

\[ \mathsf{genuine\_flip}: H_{X,B} \to \{0,1,\bot\} \]
\[ \mathsf{genuine\_flip}(e) = \left\{ \begin{array}{ll}
	\textrm{lowest order bit} & \\
	\textrm{of the full signature} & \textrm{if $\mathsf{count\_shares}(e) \geq \frac{N}{3}$} \\
	\bot & \textrm{otherwise}
\end{array}\right. \]

Let us denote the set of stages with coin value fixed to 1 as $C_1$, the set of stages with coin
value fixed to 0 as $C_0$, and the set of genuine flip stages as $C_f$. The sets satisfy $C_1 \cup
C_0 \cup C_f = \mathbb{N}$, $C_1 \cap C_0 = C_1 \cap C_f = C_0 \cap C_f = \varnothing$, and
$|C_f| = |\mathbb{N}|$. Then, the full coin flip will be defined as follows:

\[ \mathsf{coin\_flip}: H_{X,B} \to \{0,1,\bot\} \]
\[ \mathsf{coin\_flip}(e) = \left\{ \begin{array}{ll}
	1 & \textrm{if $\mathsf{stage}(e) \in C_1$} \\
	0 & \textrm{if $\mathsf{stage}(e) \in C_0$} \\
	\mathsf{genuine\_flip}(e) & \textrm{if $\mathsf{stage}(e) \in C_f$}
\end{array}\right. \]

This is all we need to reach consensus on the meta-votes.

\subsubsection{Agreement about the next block}\label{nextblock}

Using the above, every node can calculate the results of all meta-elections. Once the results are
known, they can be used to determine the next block in the ordered set.

Let us remember that the meta-elections started with a set of observers - a set of events that all
strongly see a supermajority of interesting events. The results of the meta-elections tell us which
interesting events are to be taken into account.

The properties of meta-elections ensure that all nodes will agree on the considered set of nodes.
What we need to do is change that into an agreement on what the next block should be. This is
pretty trivial, although we must consider two issues: every node could create multiple interesting
events, and every interesting event could contain multiple interesting payloads.

To counter the first issue, we can just take the earliest interesting event created by a given node.
The events created by a single node form a linear sequence, so the earliest one is well-defined.
This narrows the considered set down to a single interesting event per node.

The next step is to choose a valid block among potentially multiple ones seen by the chosen
interesting event. To do that, we can take the lexicographically first one, or use really any method
that will always choose the same element of a set.

Once we have one vote on a block per node, we just count them and the next agreed block will be the
one with the most votes. Any ties can be broken again by lexicographic ordering, or some other
method.

It is also possible to repeat the steps above for other interesting payloads of the interesting
events that have been selected by the meta-election as an optimisation, so that a single
meta-election results in appending multiple stable blocks.

This completes the description of the algorithm. The next section will prove that it is correct,
that is, that it provides robust consensus in an asynchronous setting, and in the presence of
Byzantine faults.

\section{Proof of correctness}

Let us begin by stating two important properties of the gossip graph.

\begin{defn}
	We call two gossip graphs \emph{consistent} iff for every gossip event $x$ that is present in
	both graphs, both contain the same set of ancestors of $x$ with the same sets of edges between
	them.
\end{defn}

\begin{lem}\label{consistency}
	All nodes in the network have consistent gossip graphs.
\end{lem}

\begin{lem}\label{stronglysee}
	If a pair of gossip events $(x,y)$ is a fork, and another gossip event $z$ strongly sees $x$,
	then no other gossip event in a consistent graph can strongly see $y$.
\end{lem}

We won't prove the above lemmas - they have been proved in \cite{baird2016swirlds} (as Lemma 5.11 and 5.12,
respectively). Note that lemma \ref{stronglysee} only holds if $N > 3f$, but we assume that anyway.

Let us now prove some properties of our approach stemming from it being an adaptation of
\cite{aba}.

\begin{lem}[Interesting events]\label{interesting}
	If a correct node creates an interesting event with payload $p$, then all correct nodes will
	eventually create an interesting event with payload $p$.
\end{lem}

\begin{proof}
	Let $e$ be the event created by a correct node that has interesting payload $p$. By lemma
	\ref{descendants}, eventually every correct node will create a descendant of $e$. By the
	properties of interesting payloads, either this descendant or one of its self-ancestors will
	then have interesting payload $p$. Thus, for every correct node, there will be an event that is
	an interesting event with payload $p$, which completes the proof.
\end{proof}

\begin{lem}[Aux values]\label{aux}
	If all correct nodes created an event in stage $s$, then all correct nodes will eventually
	create an event with an $aux$ value in stage $s$.
\end{lem}

\begin{proof}
	Every event in stage $s$ has at least one estimate. There is a supermajority of correct nodes,
	so there will exist a value $v$ such that at least $N/3$ correct nodes have $v$ in estimates.
	Thus, there will exist a value $v$ that will eventually be estimated by all honest nodes, which
	means it will get promoted to bin-values by all honest nodes. Once $\mathsf{bin\_values}$ is
	not empty for an event, this event also has an aux value. Since all honest nodes will
	eventually create events with non-empty bin-values, these events will have aux values, which
	completes the proof.
\end{proof}

\begin{lem}[Progress]\label{progress}
	If a correct node created a gossip event in stage $s$, every other correct node will eventually
	create an event in stage $s$ as well.
\end{lem}

\begin{proof}
	Assume $s = 0$. The first gossip event in stage 0 is an observer. If a correct node created an
	observer, it must have strongly seen a supermajority of interesting events. A supermajority
	always contains a correct node, so at least one correct node created an interesting event. By
	lemma \ref{interesting}, all correct nodes will eventually have created interesting events.

	If all correct nodes created interesting events, it means that eventually all correct nodes will
	create an event strongly seeing a supermajority of interesting events - as there is a
	supermajority of correct nodes, they continue gossipping and they never fork. Thus, all correct
	nodes will eventually create observers, which completes the proof for $s = 0$.

	Assume the lemma holds for stage $s$. We will now prove that this implies it holds for stage
	$s+1$.

	Assume a correct node created an event in stage $s+1$. This means that this event sees a
	supermajority of events in stage $s$ with some \emph{aux} values. This means that at least one
	honest node created an event in stage $s$, so by our assumption, all honest nodes will have
	eventually created an event in stage $s$. By lemma \ref{aux}, this means that all honest nodes
	will eventually have created an event in stage $s$ with an aux value. All honest nodes will
	eventually create events that see all these events with aux values, which constitutes a
	supermajority, which is enough to progress to the next stage - so all honest nodes will create
	an event in stage $s+1$. By induction, the proof is complete.
\end{proof}

\begin{lem}\label{claima}
	If all correct nodes created events in stage $s$ and stage $s$ is a genuine flip stage, then the
	estimates of the nodes' events in the next stage will be in agreement with probability
	$\geq \frac{1}{2}$.
\end{lem}

\begin{proof}
	Let us consider the worst case scenario, in which there are $f$ malicious nodes among the
	$N = 3f + 1$ nodes. Let us also assume that the adversary controls the timing of the messages,
	so by controlling which messages are delivered when, they can control the gossip pattern and
	effectively, to some extent, the values associated with the gossip events.

	Assume the adversary tries to force a disagreement among the honest nodes. The only way to do so
	is to make some honest nodes see no agreeing supermajority among the aux values, which will make
	them take the coin value as the next estimate, and other honest nodes to see a supermajority
	of aux values opposite to the coin value. In other cases the honest nodes will automatically
	have agreeing estimates in stage $s+1$.

	The adversary cannot control the coin value, so they need to learn its value first. It is only
	possible if at least one honest node published its coin share, which means it has already seen
	a supermajority of aux values. Then, for the first $2f+1$ aux values seen by the correct node,
	exactly one of the following is true:

	\begin{itemize}
		\item There are at least $f+1$ \emph{true} aux values.
		\item There are at least $f+1$ \emph{false} aux values.
	\end{itemize}

	Whichever one is true, no matter what control the adversary has over the remaining aux values,
	it cannot make other honest nodes see a supermajority for the opposite value. Thus, it is out
	of the adversary's control to make the nodes disagree, as they couldn't have known beforehand
	which value they need to have a supermajority of.

	If there is no supermajority of agreeing aux values, the honest nodes will automatically be in
	agreement. If there is one and all of them see it, they will also be in agreement. If not all
	of them see it, there is a $\frac{1}{2}$ probability that the coin value will agree, thus also
	making all honest nodes agree.

	Thus, the probability of the honest nodes agreeing in stage $s+1$ is at least $\frac{1}{2}$.
\end{proof}

\begin{lem}\label{claimb}
	If all correct nodes' first events in stage $s$ had $\mathsf{estimates} = \{v\}$, either they
	will decide $v$ in stage $s$, or their first events in stage $s+1$ will also have
	$\mathsf{estimates} = \{v\}$.
\end{lem}

\begin{proof}
	If all correct nodes only estimate $v$, there is no way for any event to see even $\frac{N}{3}$
	of estimates for $\neg v$ - so no event by a correct node will have it in its estimates in
	stage $s$.

	For a value to be an element of $\mathsf{bin\_values}$, there must be a supermajority of events
	estimating that value. Because of the above, the only value that can have a supermajority is
	$v$. Thus, every event with nonempty $\mathsf{bin\_values}$ will have it equal to $\{v\}$.
	Hence, every event with an $\mathsf{aux}$ value will have it equal to $v$.

	In order to proceed to the next stage, an event has to see a supermajority of valid
	$\mathsf{aux}$ values. No event can have a value other than $v$ as $\mathsf{aux}$ in stage $s$,
	so there will always be a supermajority for $v$. Depending on the coin flip value, this can
	either lead to deciding $v$, or estimating $v$ in stage $s+1$. Either way, the agreement will
	still hold.
\end{proof}

\begin{lem}\label{decide}
	If all correct nodes' first events in stage $s$ had $\mathsf{estimates} = \{v\}$, they will all
	decide $v$ eventually.
\end{lem}

\begin{proof}
	No matter what malicious nodes do, there is less than a third of them, so no event by a correct
	node will have $\neg v$ in estimates (by definition of the $\mathsf{est}$ function). This means
	that for $\mathsf{bin\_values}$ of an event to be non-empty, it must see a supermajority of
	estimates for $v$, as there will never be a supermajority for $\neg v$.

	The above means that no correct node will add $\neg v$ to $\mathsf{bin\_values}$, so all of
	them will eventually create an event with $\mathsf{aux} = v$. This means there will be a
	supermajority of events by different creators with $\mathsf{aux} = v$, which will make the
	correct nodes either decide in stage $s$ (if $\mathsf{coin\_flip} = v$), or estimate $v$ for the
	next stage. This will repeat until $\mathsf{coin\_flip} = v$ and the nodes decide $v$. Since
	there will be infinitely many genuine flips with the number of stages tending to infinity, and
	each genuine flip will result in $v$ with probability $\frac{1}{2}$, this will eventually happen
	with probability 1.
\end{proof}

\begin{thm}[Binary Byzantine Consensus]\label{binconsensus}
	The algorithm for calculating meta-election results presented in this paper satisfies the
	general properties of a Byzantine fault tolerant consensus algorithm:
	\begin{itemize}
		\item \emph{Validity} - if a correct node decides on a value, it has been proposed by a
			correct node.
		\item \emph{Agreement} - if a correct node decides on a value, all correct nodes decide
			on that value.
		\item \emph{Integrity} - once a correct node decides on a value, it never decides on
			another value.
		\item \emph{Termination} - all correct nodes eventually decide with probability 1.
	\end{itemize}
\end{thm}

\begin{proof}[Validity]
	We will prove an equivalent statement: that if initially all correct nodes propose $v$, then
	all correct nodes will decide $v$. Since $v$ is a binary value, a node can only decide a value
	not proposed by a correct node if all correct nodes propose $v$, and the node decides $\neg v$.
	Thus, deciding $v$ when all correct nodes propose $v$ is equivalent to always deciding on a
	value proposed by a correct node.

	If all correct nodes propose $v$, they will all put $v$ in their estimates. By Lemma
	\ref{decide}, they will all decide $v$ eventually.
\end{proof}

\begin{proof}[Agreement]
	Assume there is an event $e$ created by a correct node which was able to decide a value $v$. It
	means that $\mathsf{coin\_flip}(e) = v$ and this event must have seen a supermajority of events
	with $\mathsf{aux} = v$. This means there was no supermajority for $\neg v$. Thus, if a correct
	node has seen a supermajority in this stage, it must have been for $v$, so it would decide $v$.
	If it hasn't, it would estimate $v$ for the next stage, which means there will be agreement at
	the start of the next stage. Following Lemma \ref{claimb}, this agreement will propagate to the
	end of the stage and the next stages, until everyone decides $v$.
\end{proof}

\begin{proof}[Integrity]
	Once an event $e$ created by a correct node decides on a value $v$, all later events created by
	that node will have event $e$ as an ancestor. Following the definition of
	$\mathsf{meta\_election}$, all later events will also decide $v$.
\end{proof}

\begin{proof}[Termination]
	By Lemma \ref{progress}, if a correct node creates an event in stage $s$, then every correct
	node eventually creates an event in stage $s$. This means there will be events by
	$> \frac{2}{3}N$ correct nodes, which will eventually be seen by every correct node.  Every
	such event will have non-empty estimates. It is not possible for both 0 and 1 to be estimated
	by $< \frac{N}{3}$ events by different correct nodes, so at least one of those values will
	eventually be an element of estimates of every correct node's event.

	Eventually, the events with agreeing estimates will all be seen by an event created by every
	correct node. Hence, every correct node will eventually create an event with non-empty
	$\mathsf{bin\_values}$, and so an $\mathsf{aux}$ value.

	The events with $\mathsf{aux}$ values will eventually be seen by every correct node's event,
	which means every correct node will eventually either decide or progress to the next stage.

	By Lemma \ref{claima}, after every stage with a genuine coin flip, all nodes' estimates agree
	with probability $> \frac{1}{2}$. This means that the probability of estimates still not
	agreeing at stage $s$ is less than:
	\[ (1 - \frac{1}{2})^{\mathsf{gf}(s)} = \frac{1}{2^{\mathsf{gf}(s)}} \]
	where $\mathsf{gf}(s)$ is the number of genuine flips up to stage $s$. Since we assume an
	infinite number of genuine flip stages as $s$ tends to infinity, $\mathsf{gf}(s)$ grows to
	infinity as $s$ grows, which implies that the probability of estimates not agreeing tends to 0.
	This means that the estimates will eventually agree with probability 1. By Lemma \ref{decide},
	the nodes will decide eventually after that happens.
\end{proof}

The above theorem proves that our algorithm will reach agreement about every single meta-election
in a Byzantine fault tolerant way. This is not the end, though - we also need to prove that
meta-elections lead to agreement about the next block in the ordered set. The proof of that is
presented below.

\begin{lem}\label{minvotes}
	If the result of a meta-election is $v$, there have been at least $\frac{N}{3}$ meta-votes for
	$v$.
\end{lem}

\begin{proof}
	Assume there have been less than $\frac{N}{3}$ meta-votes for $v$ and $v$ has been decided.
	When nodes that initially meta-voted $v$ create an event that sees a supermajority of
	meta-votes, this supermajority must contain at least $\frac{N}{3}$ votes for $\neg v$ - so
	their estimates will contain $\neg v$. On the other hand, no node that meta-voted $\neg v$ can
	ever create an event that will see at least $\frac{N}{3}$ estimates for $v$, so they won't add
	$v$ to estimates.

	Due to the above, any supermajority among the estimates must be for $\neg v$. Any event with
	non-empty $\mathsf{bin\_values}$ can thus only have $\neg v$ in this set, which means that
	all valid $\mathsf{aux}$ values will also be $\neg v$, which will lead to a decision on $\neg
	v$ eventually.

	This is a contradiction. Such a situation is impossible, which proves the lemma.
\end{proof}

\begin{lem}\label{nonempty}
	The set of nodes for which the result of meta-election is $true$ is always non-empty.
\end{lem}

\begin{proof}
	Assume all meta-elections resulted in $false$. By Lemma \ref{minvotes}, at least $\frac{N}{3}$
	nodes meta-voted $false$ for every node, so there have been at least $\frac{N^2}{3}$ meta-votes
	for $false$.

	On the other hand, by definition of an observer, every node voted $true$ for more than
	$\frac{2}{3}N$ nodes - so there have been more than $\frac{2}{3}N^2$ meta-votes for $true$,
	which leaves less than $\frac{N^2}{3}$ meta-votes for $false$ (there are $N^2$ meta-votes in
	total: $N$ nodes meta-vote in $N$ meta-elections).

	This is a contradiction, which proves the lemma.
\end{proof}

\begin{thm}[Byzantine Consensus]
	The algorithm for calculating the next block presented in this paper satisfies the general
	properties of a Byzantine fault tolerant consensus algorithm:
	\begin{itemize}
		\item \emph{Validity} - if a correct node decides on a next block, its payload was in at
			least one interesting event created by a correct node.
		\item \emph{Agreement} - if a correct node decides on a next block, all correct nodes decide
			on that block.
		\item \emph{Integrity} - once a correct node decides on a next block, it never decides on
			another block.
		\item \emph{Termination} - all correct nodes eventually decide with probability 1.
	\end{itemize}
\end{thm}

\begin{proof}[Validity]
	Assume that no correct node created an interesting event with payload $p$, but $p$ was still
	decided as the payload of the next block.

	This means that only faulty nodes could create interesting events with payload $p$, so there is
	less than $N/3$ such interesting events. Furthermore, if a correct node created a descendant of
	such an interesting event, then either it or one of its self-ancestors would have to be an
	interesting event with payload $p$ as well. This would contradict our assumption, so no correct
	node could have created a descendant of an interesting event with payload $p$.

	It follows from this that when correct nodes create observers, no such observer can meta-vote
	\emph{true} for a creator of an interesting event with payload $p$ (as it would require the
	observer to be a descendant of such an event). This means that after the meta-election is
	complete, no interesting event with payload $p$ will be taken into consideration, so $p$ will
	not be decided as the payload of the next stable block.

	QED by contradiction.
\end{proof}

\begin{proof}[Agreement]
	Assume that a correct node decided the next block $B$. Since a decision has been reached, this
	means that there is consensus about the meta-votes, so every correct node will have chosen the
	same nodes' interesting events.

	For any observer to meta-vote $true$ on a node, it must have strongly seen an interesting event
	by that node. By Lemma \ref{stronglysee}, even if that node created a fork, if any other
	observers also voted $true$ on that node, they must have strongly seen interesting events on the
	same fork. Thus, we can consider interesting events by all elected nodes to form linear
	histories - which will be seen the same way by all nodes by Lemma \ref{consistency}.

	In a linear history, the earliest interesting event is well-defined. Also, because all correct
	nodes see the same history, they will all choose the same interesting event as the earliest. If
	the interesting events has multiple payloads, all correct nodes will use the same tie-breaker
	algorithm and choose the same single one. Thus, all correct nodes will gather the same set of
	votes, and because they use the same voting rules, decide the same block as the next one.
\end{proof}

\begin{proof}[Integrity]
	By construction of the algorithm, once the next block has been decided, it is appended to the
	ordered set and no other block can be decided in its place.
\end{proof}

\begin{proof}[Termination]
	The consensus algorithm starts when a correct node creates an interesting event. Once that
	happens, by lemma \ref{interesting}, all correct nodes will eventually create interesting
	events, and this in turn implies that all of them will create observers. Once there is a
	supermajority of observers, we start the binary agreement algorithm, which will terminate by
	Theorem \ref{binconsensus}. After binary agreement terminates, because the set of voters for
	the next block is non-empty (by Lemma \ref{nonempty}), the next block is already determined -
	so the agreement about the next block also terminates.
\end{proof}

\section{Conclusions}

A new consensus algorithm has been presented, building upon some previous achievements in this
field (\cite{baird2016swirlds}, \cite{aba}), but combining their features in a novel way. It works under
asynchronous conditions, uses a gossip graph (like \cite{gkagol2018aleph}, \cite{baird2016swirlds} and \cite{rocket2018snowflake}),
and a common coin. It is also leaderless (barring the very initialisation of the network, when a
trusted dealer may be required) and an open source implementation is provided. We believe this
approach will be useful in numerous applications, one of which is the SAFE Network.

\appendix

\titlelabel{\thesection\par}

\section{: Extending the algorithm to a network with dynamic membership}

The main algorithm is formulated in terms of a network in which all the members are known a priori
(a static network). This is enough in some settings, but sometimes it is necessary to allow
the set of members to be modified, so that members of the network can join and leave at will.

In order to accommodate dynamic membership in the network, every node has to keep a record of who the
current members are. We will call this record the \emph{network members list}. This list is
initialised with the so called \emph{genesis group} and can only be modified as a result of a block
becoming stable.

Changing the membership list requires re-generating the common coin secret key shares, too. We need
to generate them in a way that doesn't allow any single node to get to know more than just their
share. Fortunately, there are distributed key generation (DKG) algorithms in existence that solve
this problem. One of the simplest ones is \cite{dkg}, but it requires synchronous communications.

Fortunately, we can simulate synchronous communications using the instances of PARSEC held by the
old set of members. DKG messages can be input as votes into the graph, and the consensus algorithm
will ensure that all nodes will process them in the same order. The messages in the order that was
agreed upon can also be passed to the nodes that are joining, thus allowing them to generate their
key shares.

In summary, a membership change would be processed as follows:

\begin{enumerate}
	\item The old set of members votes for a membership change (adding or removing a node).
	\item A block with node addition/removal becomes stable.
	\item Nodes from the old members set start the DKG algorithm and begin voting for DKG messages.
	\item Blocks with DKG messages become stable. At some point, enough of them are stable to
		complete the DKG algorithm.
	\item Nodes from the old set calculate their new key shares and update their members lists.
	\item If a new node was joining, it will start receiving gossip containing the gossip events
		since genesis up to this point. By processing the graph, it can learn of all the blocks that
		became stable, including the DKG messages, from which it will be able to derive its own key
		share. It will also arrive at the current members list.
	\item The membership change is complete.
\end{enumerate}

This method ensures that every meta-election uses a constant list of members from start to finish - the members list only gets modified once a meta-election finishes (the one regarding the block that completed the DKG), and another one is not yet started. Thanks to this approach, the proofs of correctness apply without modifications.

\label{references}

\bibliographystyle{plain}
\bibliography{whitepaper}

\end{document}